\RequirePackage{amsmath}
\documentclass[runningheads]{llncs}

\usepackage[T1]{fontenc}
\usepackage{graphicx}
\usepackage{booktabs}
\usepackage[misc]{ifsym}
\newcommand{\corr}{(\Letter)}

\usepackage{amssymb}
\usepackage{url}
\usepackage{color}

\usepackage{cite}
\usepackage[colorlinks=true, linkcolor=blue, citecolor=blue, urlcolor=blue]{hyperref}
\usepackage{multirow}

\usepackage{algorithm}
\usepackage{algorithmicx}
\usepackage{algpseudocode}

\DeclareMathOperator*{\argmax}{arg\,max}

\begin{document}

\title{Multi-Objective Submodular Maximization with Differential Privacy}

\titlerunning{Multi-Objective Submodular Maximization with Differential Privacy}

\author{Ting Hou\inst{1} \and Yanhao Wang\inst{1}\orcidID{0000-0002-7661-3917} \corr \and Yiping Wang\inst{2} \and Cen Chen\inst{1} \and Minghao Zhao\inst{1} \and Fan Dang\inst{3}}

\authorrunning{T. Hou et al.}

\institute{East China Normal University, Shanghai, China\\
\email{51265903077@stu.ecnu.edu.cn} \; \email{\{yhwang,cenchen,mhzhao\}@dase.ecnu.edu.cn}
\and
The Chinese University of Hong Kong, Hong Kong, SAR, China\\
\email{ypwang@se.cuhk.edu.hk}
\and
Beijing Jiaotong University, Beijing, China\\
\email{dangfan@bjtu.edu.cn}}

\maketitle

\begin{abstract}
In this paper, we study multi-objective submodular maximization (MOSM) subject to a cardinality constraint under differential privacy (DP). Specifically, we aim to select a set of at most $k \in \mathbb{Z}_{+}$ elements to maximize the minimum of $d > 1$ monotone submodular functions while satisfying $\varepsilon$-DP. Although extensive studies have been conducted on both differentially private single-objective submodular maximization on sensitive data and non-private MOSM, to the best of our knowledge, there has not yet been any prior work on MOSM with DP. We propose two novel algorithms: the first extends the classic greedy algorithm and the second employs a truncation technique, both of which are integrated with DP mechanisms for privacy protection and achieve approximation guarantees for MOSM. Finally, we conduct numerical experiments on two submodular maximization applications, namely maximum coverage and facility location, in multi-objective settings to validate the efficacy and efficiency of our proposed algorithms.

\end{abstract}

\section{Introduction}
\label{sec:intro}

The notion of submodularity captures the ``diminishing returns'' property of set functions; that is, the marginal gain of adding any new element to a smaller set is always not less than that of adding it to a larger set.
This notion arises naturally in many machine learning and data mining problems, including data summarization \cite{LindgrenWD17}, feature selection \cite{BaoHZ22}, influence maximization \cite{KempeKT03}, sensor placement \cite{krause2008robust}, etc.
Due to its broad applications, submodular function maximization has been regarded as a fundamental combinatorial optimization problem and has attracted much attention over the past decades.
Since many classes of submodular maximization (SM) problems are NP-hard \cite{Feige98}, any polynomial-time algorithm can only provide approximate solutions for them unless $P = NP$ \cite{KrauseG14}.
In particular, Nemhauser et al.~\cite{NemhauserWF78} first proposed a greedy algorithm that achieves the best possible approximation factor of $1 - 1/e$ to maximize a monotone submodular function subject to a cardinality constraint.
Following this seminal work, many methods \cite{KrauseG14} have been proposed to design efficient approximation algorithms for submodular maximization.

Classic submodular maximization problems consider only a single objective function.
However, in many real-world scenarios, we need to optimize multiple submodular objective functions simultaneously.
As an illustrative example, in a spatial social network with multiple communities, the fairness-aware facility location problem \cite{WangLBW24} aims to select a subset of $k$ locations to deploy facilities (e.g., hospitals and schools) so that the benefits are distributed equally across all communities.
This can be formulated as a \emph{multi-objective submodular maximization} (MOSM) problem, in which the benefit received by each community is represented by a submodular utility function, and the minimum utility across all communities should be maximized to improve the conditions of the least well-off community.
As indicated by Krause et al.~\cite{krause2008robust}, MOSM is NP-hard and cannot be approximated with any constant factor unless $P = NP$.
Existing algorithms for MOSM achieve either constant approximation factors in special cases when the number of objective functions is limited \cite{OrlinSU18, Udwani18, FuB0P21} or weaker bicriteria approximations \cite{krause2008robust, TorricoSPHNA19}, which provide approximate solutions by violating the specified constraints (e.g., providing more than $k$ elements for a cardinality constraint $k$).

Furthermore, existing methods for submodular maximization typically adopt a \emph{value oracle} model, which assumes that there is an oracle to evaluate the value of a submodular function w.r.t.~any set of elements based on public data.
However, in certain applications, the values of submodular functions must be evaluated on sensitive, private data about individuals.
Continuing with the example of facility location, the benefits of the locations of selected facilities should be evaluated based on their distances to users, which requires sensitive location information about individuals.
Due to legal and ethical reasons, the selection of locations for facility deployment should not compromise the privacy of any individual.
Nevertheless, the theory of submodular maximization itself cannot provide any privacy guarantee.
Toward this end, the concept of differential privacy \cite{DworkR14} (DP), the de facto standard for privacy-preserving computing, has been introduced into submodular maximization problems.
By adding random noise to submodular maximization algorithms, we can obtain a differential privacy guarantee, i.e., no individual's data can be inferred from the computation procedure, since the algorithms' outputs are indistinguishable for any two datasets $D$ and $D'$ that differ by a single tuple.
In this case, a privacy parameter $\varepsilon > 0$ quantifies the degree of indistinguishability: lower values of $\varepsilon$ indicate greater indistinguishability and thus stronger privacy protection.
Although there have been algorithms for submodular maximization with DP \cite{GuptaLMRT10,MitrovicB0K17,RafieyY20,SunLZZ22,Perez-SalazarC21,ChaturvediNZ21,SadeghiF21,ChaturvediNN23,GuoLXY23,HuXDM24,WangZCW24,Rafiey24}, they are all specific to the single-objective setting.
To the best of our knowledge, there has been no prior work on MOSM with DP.

In this paper, we first systematically study the problem of maximizing the minimum of $d$ monotone submodular functions $f_1, \dots, f_d$ under a cardinality constraint $k$ while satisfying $\varepsilon$-DP ($\varepsilon > 0$).
Our main contributions are summarized as follows.
\begin{itemize}
    \item We propose \textsc{DP-MultiGreedy}, a new differentially private algorithm for MOSM that extends the classic greedy algorithm \cite{NemhauserWF78}.
    \textsc{DP-MultiGreedy} proceeds in two phases.
    In the first phase, it invokes \textsc{DP-Greedy} \cite{MitrovicB0K17}, which utilizes the exponential mechanism \cite{McSherryT07}, to obtain $d$ solutions $S_1, \dots, S_d$, each of size $\lfloor \tfrac{k}{d} \rfloor$, for maximizing the objective functions $f_1, \dots, f_d$ independently.
    Then, starting from $S = \bigcup_{j=1}^d S_j$, the second phase iteratively includes new elements from $V \setminus S$ to maximize $\min_{j} f_j(S)$ until $|S| = k$.
    Theoretically, \textsc{DP-MultiGreedy} returns a solution $S$ such that
    $$
        F(S) \geq \left(1 - e^{-\frac{\lfloor k/d \rfloor}{k}}\right) F(OPT) - O\!\left( \frac{d k}{\varepsilon} \log\!\left(\frac{k n}{\eta}\right) \right)
    $$
    with probability at least $1 - \eta$, where $F(S) = \min_{j \in [d]} f_j(S)$, $n = |V|$, and $OPT$ is the optimal solution for MOSM.
    Furthermore, it runs in $O(nmk)$ time, where $m = |D|$, and satisfies $\varepsilon$-DP.

    \item We propose \textsc{DP-Bicriteria}, another differentially private algorithm for MOSM.
    \textsc{DP-Bicriteria} extends \textsc{Saturate} \cite{krause2008robust} in the context of DP, which adopts the truncation technique to transform an instance of MOSM into multiple single-objective instances.
    By solving each instance with \textsc{DP-Greedy} \cite{MitrovicB0K17} and choosing the best solution, it provides a solution $S$ that achieves a bicriteria approximation for the original MOSM problem; that is,
    $$
        F(S) \geq (1 - \alpha) F(OPT) - O\big( \frac{d k \log{m}}{\varepsilon} \log\big(\frac{k n}{\eta} \log{m} \log(\frac{d}{\alpha})\big) \big)
    $$
    with probability at least $1 - \eta$ and $|S| = O(k \frac{\log{d}}{\alpha})$.
    \textsc{DP-Bicriteria} runs in $O\big(n m d k \log{m}$ $\log(\frac{d}{\alpha})\big)$ time and also satisfies $\varepsilon$-DP.

    \item We evaluate our proposed algorithms with numerical experiments on two submodular maximization applications, namely \textit{facility location} and \textit{maximum coverage}, in the multi-objective setting.
    The results show that they offer good privacy-utility tradeoffs. Moreover, \textsc{DP-MultiGreedy} achieves better performance when $d = 2$, whereas \textsc{DP-Bicriteria} outperforms \textsc{DP-MultiGreedy} for a larger $d$ (e.g., $d=5$).
\end{itemize}

\section{Related Work}
\label{sec:lit}

\noindent
\textbf{Multi-Objective Submodular Maximization:}
Recently, a lot of effort has been devoted to submodular maximization problems in multi-objective settings.
Krause et al. \cite{krause2008robust} first investigated the MOSM problem of maximizing the minimum of multiple monotone submodular functions with a cardinality constraint.
They also showed the inapproximability and proposed a bicriteria approximation algorithm.
Orlin et al. \cite{OrlinSU18} proposed a $(1 - 1/e - \alpha)$-approximation algorithm for MOSM with only $d = 2$ objective functions.
Torrico et al. \cite{TorricoSPHNA19} further proposed a bicriteria approximation algorithm for MOSM with matroid constraints.
Udwani \cite{Udwani18} and Fu et al. \cite{FuB0P21} proposed multiplicative weight updates (MWU) frameworks for MOSM that achieve an approximation factor of $(1 - 1/e)^2 - \alpha$ when $d = o(\tfrac{k}{\log^3{k}})$.
Although these methods use the same MOSM formulation as ours, they do not take data privacy issues into account.

Other MOSM studies adopted different formulations from that in~\cite{krause2008robust}.
Malherbe and Scaman \cite{MalherbeS22} proposed to maximize the $\ell$-th quantile of the values of multiple submodular functions for $\ell \in [0, 1]$, where $\ell = 0$ corresponds to the MOSM problem and $\ell = 1$ corresponds to vanilla submodular maximization.
Moreover, several studies \cite{OhsakaT21, GershteinMY21, WangLBW24, FazzoneWB24} considered maximizing a submodular function subject to a cover constraint that the value(s) of one or more submodular functions should reach a given threshold.
There were also some investigations \cite{SomaY17a, Feng021, WangZM23} on minimizing the regret ratio w.r.t.~the weighted sum of multiple submodular functions.
The algorithms proposed for these variants cannot be directly used for our problem.
Meanwhile, they also ignore data privacy.

\noindent
\textbf{Differentially Private Submodular Optimization:}
A number of studies have been performed on differentially private submodular optimization in recent years.
Gupta et al. \cite{GuptaLMRT10} first extended the classic greedy algorithm \cite{NemhauserWF78} to maximize a decomposable, monotone, and submodular function subject to a cardinality constraint, replacing the greedy selection in each iteration with the exponential mechanism \cite{McSherryT07} to satisfy differential privacy.
Mitrovic et al. \cite{MitrovicB0K17} generalized Gupta et al.'s results to non-decomposable, non-monotone submodular functions and other constraints beyond cardinality.
Cardoso et al. \cite{CardosoC19} and Perez-Salazar et al. \cite{Perez-SalazarC21} studied differentially private online submodular minimization and maximization problems.
Continuous greedy algorithms have also been extended to differentially private submodular maximization problems \cite{RafieyY20, ChaturvediNZ21, SadeghiF21, SunLZZ22}.
Chaturvedi et al. \cite{ChaturvediNN23} and Guo et al. \cite{GuoLXY23} proposed streaming algorithms for submodular maximization with differential privacy.
Hu et al. \cite{HuXDM24} investigated differentially private submodular maximization over integer lattices.
Wang et al. \cite{WangZCW24} and Rafiey \cite{Rafiey24} focused on submodular maximization with differential privacy in federated settings.
However, the above methods are all specific to single-objective submodular optimization problems and cannot handle multiple objectives.

\section{Preliminaries}
\label{sec:def}

\noindent
\textbf{Multi-Objective Submodular Maximization:}
First, we use $[n]$ to denote a sequence of integers from $1$ to $n$.
Let $V$ be a finite ground set of $n$ elements, and $X$ be a finite data universe.
A dataset $D$ is a $m$-tuple set $D = \{x_1, \dots, x_m\} \in X^m$.
Then, we define a set function $f: 2^V \mapsto \mathbb{R}_{\geq 0}$ that maps any set $S \subseteq V$ to its nonnegative utility value $f(S | D)$ for~$D$.
The function $f$ is said to be \emph{submodular} if for all $S \subseteq T \subseteq V$ and each element $v \in V \setminus T$, $f(S \cup \{v\} | D) - f(S | D) \geq f(T \cup \{v\} | D) - f(T | D)$.
Moreover, the function $f$ is \emph{monotone} if $f(S | D) \leq f(T | D)$ for any $S \subseteq T \subseteq V$.
We assume that $f$ is normalized, i.e., $f(\emptyset | D) = 0$, monotone, and submodular.
Typically, we consider that evaluating the value of $f(S | D)$ for any $S \subseteq V$ takes $O(m)$ time.
We aim to maximize the function $f$ subject to a constraint indicated by a collection $\mathcal{I} \subseteq 2^{V}$ of feasible solutions, i.e., $\max_{S \in \mathcal{I}} f(S | D)$.
In this paper, we focus on the most straightforward case when $\mathcal{I}$ is defined by a cardinality constraint $k \in \mathbb{Z}_{+}$, i.e., $\mathcal{I} = \{S \subseteq V: |S| \leq k\}$ to restrict the size of any feasible solution to at most $k$.
When the context is clear, we will drop $D$ from the notation of $f$.

We investigate the problem of multiple-objective submodular maximization (MOSM) with a cardinality constraint.
Specifically, suppose that $f_1, \dots, f_d$ are $d$ normalized, monotone, and submodular functions defined on the same ground set $V$ and for the same dataset $D$.
We aim to maximize the minimum of $f_1, \dots, f_d$ under a cardinality constraint $k$, i.e.,
\begin{equation}\label{eq:mosm}
    OPT = \argmax_{S \subseteq V : |S| \leq k} \min_{i \in [d]} f_i(S).
\end{equation}
We use $F(S) = \min_{i} f_i(S)$ to denote the minimum value function.

\noindent
\textbf{Differential Privacy:}
We introduce several fundamental concepts regarding differential privacy (DP).
For any two datasets $D, D' \in X^{m}$, we say that they are neighboring, denoted as $D \sim D'$, if they differ in at most one tuple.
A randomized algorithm $\mathcal{M} : D \rightarrow \mathcal{R}$ is $\varepsilon$-differentially private \cite{DworkR14} if, for any pair of neighboring datasets $D \sim D' \in X^{m}$ and for any measurable output set $R \subseteq \mathcal{R}$,
\begin{equation}\label{eq:dp}
    \Pr\left[\mathcal{M}(D) \in R\right] \leq e^{\varepsilon} \cdot \Pr\left[\mathcal{M}(D') \in R\right].
\end{equation}
An important property of the function of interest for differentially private algorithms is that its function value does not change significantly w.r.t.~small changes in the input dataset, which is formalized by \emph{sensitivity}.
For a set function $f$, a collection $\mathcal{I}$ of feasible solutions, and any pair of neighboring datasets $D \sim D'$, the sensitivity \cite{MitrovicB0K17} is defined as
\begin{equation}\label{eq:sen}
     \Delta f = \max_{D \sim D'} \max_{S \in \mathcal{I}} \big| f(S | D) - f(S | D') \big|.
\end{equation}
We assume that each function $f_i$ has a constant sensitivity $\Delta f_i = O(1)$.
W.l.o.g., we align the sensitivities of all $d$ functions to their maximum, i.e., $\Delta = \max_{i} \Delta f_i$.

Furthermore, our analysis is based on the composition theorem \cite{DworkR14} for DP, which quantifies the cumulative loss of privacy when multiple DP mechanisms are used sequentially.
Let $\mathcal{M}_1, \dots, \mathcal{M}_l$ be $l$ randomized algorithms.
Suppose $\mathcal{M}_j$ is $\varepsilon_j$-DP for each $j \in [l]$.
When $\mathcal{M}_1, \dots, \mathcal{M}_l$ are run sequentially on the same dataset $D$, the composition algorithm $\mathcal{M} = (\mathcal{M}_1, \dots, \mathcal{M}_l)$ is $\sum_{j=1}^{l} \varepsilon_j$-DP.
Next, we introduce the basic DP mechanisms that we use to build our algorithms.

\noindent
\textbf{Laplace Mechanism \cite{DworkR14}:}
The Laplace mechanism is designed for numeric queries.
It adds noise from the Laplace distribution to the query result.
For a query $q: X^{m} \mapsto \mathbb{R}$ with sensitivity $\Delta q$ that maps a dataset $D \in X^{m}$ to a real number, the Laplace mechanism is represented as $\mathcal{M}_L(q, D, \varepsilon) = q(D) + \mathrm{Lap}(0, \tfrac{\Delta q}{\varepsilon})$, where $q(D)$ denotes the result of $q$ on $D$, and $\mathrm{Lap}(0, \tfrac{\Delta q}{\varepsilon})$ is a random variable from the Laplace distribution with mean $0$ and scale $\tfrac{\Delta q}{\varepsilon}$.
According to~\cite{DworkR14}, the Laplace mechanism is $\varepsilon$-DP.
In addition, the output $\widetilde{q}(D) = \mathcal{M}_L(q, D, \varepsilon)$ satisfies that $\forall \eta \in (0, 1]$,
\begin{equation}\label{eq:lap}
    \Pr\left[\left|q(D) - \widetilde{q}(D) \right| \geq \ln(\frac{1}{\eta}) \cdot \frac{\Delta q}{\varepsilon} \right] \leq \eta.
\end{equation}

\noindent
\textbf{Exponential Mechanism \cite{McSherryT07}:}
The exponential mechanism is a primitive for selecting a high-utility element from the ground set.
We define a utility function $u: V \times X^{m} \mapsto \mathbb{R}$ to measure how good an element $v \in V$ is for the dataset $D$, and $\Delta u$ is the sensitivity of $u$.
The basic idea is to pick each element in $V$ sequentially and randomly based on its utility.
Specifically, the exponential mechanism $\mathcal{M}_{E}(D, u, V, \varepsilon)$ returns each element $v \in V$ with a probability proportional to $\exp\big(\frac{\varepsilon u(D, v)}{2\Delta u}\big)$.
According to~\cite{McSherryT07}, the exponential mechanism is $\varepsilon$-DP.
Moreover, it offers theoretical guarantees on the utility of the selected element.
Let $\mathrm{OPT}_{u}(D) = \max_{v \in V} u(D, v)$ denote the maximum utility of any element $v \in V$ w.r.t.~$D$.
Then, the output $\widetilde{v} = \mathcal{M}_{E}(D, u, V, \varepsilon)$  satisfies that $\forall \eta \in (0, 1]$,
\begin{equation}\label{eq:em}
     \Pr\left[ u(\widetilde{v}, D) \leq \mathrm{OPT}_{u}(D) - \frac{2\Delta u}{\varepsilon}\ln(\frac{|V|}{\eta}) \right] \leq \eta.
\end{equation}

We summarize a list of frequently used symbols in Table~\ref{tab-notation}.

\begin{table}[t]
    \caption{List of frequently used symbols.}
    \label{tab-notation}
    \centering
    \begin{tabular}{|c|l|}
        \hline
        \textbf{Symbol} & \textbf{Description} \\
        \hline
        $v \in V$ & A single element from the ground set\\
        $n$ & The number of elements in the ground set $V$\\
        $D \sim D'$ & A (private) dataset and its neighboring dataset\\
        $f_1, \dots, f_d$ & A set of $d$ submodular objective functions\\
        $F$ & The minimum value function of $f_1, \dots, f_d$\\
        $k$ & The cardinality constraint\\
        $S$ & A solution for the MOSM problem\\
        $OPT$ & The optimal solution for the MOSM problem\\
        $\Delta f, \Delta$ & The sensitivity of $f$ and the aligned sensitivity\\
        $\varepsilon > 0$ & The differential privacy parameter\\
        $\eta \in (0, 1)$ & The confidence parameter\\
        $\xi > 0$ & The additive noise term due to DP\\
        $\mathcal{M}_L, \mathcal{M}_E$ & The Laplace and exponential mechanisms\\
        \hline
    \end{tabular}
\end{table}

\section{Algorithms}
\label{sec:alg}

In this section, we propose two differentially private algorithms for MOSM.
 
\subsection{The DP-MultiGreedy Algorithm}
\label{subsec:alg1}

\noindent
\textbf{Algorithm Description:}
We first outline \textsc{DP-Greedy} to maximize a monotone submodular function under a cardinality constraint $k$ with differential privacy, denoted as $\mathcal{G}(f, D, V, k, \varepsilon)$, in Algorithm~\ref{lm:dpgreedy}, as this algorithm serves as a building block in our \textsc{DP-MultiGreedy} algorithm.
Starting from $S_0 = \emptyset$, \textsc{DP-Greedy} selects an element $\widetilde{v}_i \in V$ that maximizes the marginal gain function $u_i(D, v) := f(S_{i - 1} \cup \{v\}) - f(S_{i - 1})$ with respect to the current solution $S_{i - 1}$ with $\frac{\varepsilon}{k}$-DP using the exponential mechanism in the $i$-th iteration.
After $k$ iterations, $S_k$ is returned as the final solution $S$.
\textsc{DP-Greedy} satisfies $\varepsilon$-DP due to the sequential composition over $k$ iterations.

\begin{algorithm}[t]
\caption{\textsc{DP-Greedy} $\mathcal{G}$ \cite{MitrovicB0K17}}
\label{alg:greedy}
\begin{algorithmic}[1]
    \Require Ground set $V$, dataset $D$, set function $f: 2^V \mapsto \mathbb{R}$, cardinality constraint $k \geq 2$, privacy parameter $\varepsilon > 0$
    \Ensure Solution set $S \subseteq V$ with $|S| = k$
    \State Initialize $S_0 \gets \emptyset$;
    \For{$i = 1, \dots, k$}
        \State Define $u_i(D, v) = f(S_{i - 1} \cup \{v\}) - f(S_{i-1})$;
        \State Identify $\widetilde{v}_i \gets \mathcal{M}_{E}(D, u_i, V \setminus S_{i - 1}, \frac{\varepsilon}{k})$;
        \State Update $S_i \gets S_{i-1} \cup \{\widetilde{v}_i\}$;
    \EndFor
    \State \Return{$S \gets S_k$};
\end{algorithmic}
\end{algorithm}

In general, the \textsc{DP-MultiGreedy} algorithm extends \textsc{DP-Greedy} \cite{MitrovicB0K17} to the multi-objective setting.
The detailed procedure of \textsc{DP-MultiGreedy} is presented in Algorithm~\ref{alg:dp_multigreedy}.
The algorithm consists of two phases.
In the first phase, it applies \textsc{DP-Greedy} independently to each objective function $f_j$ for $j \in [d]$, under a cardinality constraint $\lfloor k/d \rfloor$ and a privacy budget $\varepsilon_1/d$.
As a result, it obtains $d$ subsets $S_1, S_2, \dots, S_d$, each of size $\lfloor k/d \rfloor$.
It then forms the initial solution $S$ by taking their union $\bigcup_{j=1}^d S_j$.
Clearly, $|S| \leq \lfloor k/d \rfloor \cdot d \leq k$.
In the second phase, it iteratively augments the current solution $S$ until $|S| = k$.
The remaining privacy budget $\varepsilon_2$ is allocated equally over iterations in this phase.
At each iteration, it first perturbs the current value of each objective function $f_j(S)$ using the Laplace mechanism with a privacy budget $\frac{\varepsilon_2}{2d(k-k'')}$ and identifies the worst-off objective accordingly.
It then defines the marginal gain function with respect to this worst-off objective and uses the exponential mechanism with a privacy budget $\frac{\varepsilon_2}{2(k-k'')}$ to privately select a new element.
The selected element is added to $S$, and the process repeats until $|S| = k$.
Each iteration consumes a privacy budget $\frac{\varepsilon_2}{k-k''}$, and by sequential composition, the second phase satisfies $\varepsilon_2$-DP.
The final solution $S$ is returned after the second phase.

\begin{algorithm}[t]
\caption{\textsc{DP-MultiGreedy}}
\label{alg:dp_multigreedy}
\begin{algorithmic}[1]
    \Require Ground set $V$, dataset $D$, set functions $f_1, \dots, f_d$, cardinality constraint $k \geq d$, privacy parameters $\varepsilon_1, \varepsilon_2$
    \Ensure Solution set $S \subseteq V$
    \State Compute $k' \gets \lfloor k/d \rfloor$;
    \For{$j = 1, \dots, d$}
        \State Compute $S_j \gets \mathcal{G}(f_j, D, V, k', \varepsilon_1/d)$;
    \EndFor
    \State Update $S \gets \bigcup_{j=1}^d S_j$ and compute $k'' = |S|$;
    \While{$|S| < k$}
        \For{$j = 1, \dots, d$}
            \State Compute $\hat{f}_j(S) \gets f_j(S) + \mathrm{Lap}\!\left(\frac{2d(k-k'')}{\varepsilon_2}\right)$;
        \EndFor
        \State Identify $j^* \gets \arg\min_{j \in [d]} \hat{f}_j(S)$;
        \State Define $u(D, v) \gets f_{j^*}(S \cup \{v\}) - f_{j^*}(S)$;
        \State Identify $\widetilde{v} \gets \mathcal{M}_{E}\!\left(D, u, V \setminus S, \frac{\varepsilon_2}{2(k-k'')}\right)$;
        \State Update $S \gets S \cup \{\widetilde{v}\}$;
    \EndWhile
    \State \Return{$S$};
\end{algorithmic}
\end{algorithm}

\noindent
\textbf{Theoretical Analysis:}
Next, we analyze the privacy, approximation, and time and space complexity of \textsc{DP-MultiGreedy} in Algorithm~\ref{alg:dp_multigreedy}.
We first extend the approximation guarantee of \textsc{DP-Greedy} from a special case of $k = l$ in~\cite{MitrovicB0K17} to a general case of arbitrary $l \geq 1$ in the following lemma.
\begin{lemma}
\label{lm:dpgreedy}
    Suppose $f: 2^{V} \mapsto \mathbb{R}_{\geq 0}$ is a monotone submodular function with sensitivity $\Delta$.
    By running \textnormal{\textsc{DP-Greedy}} $\mathcal{G}(f, D, V, l, \varepsilon)$ in $l$ iterations for any $l \geq 1$ and $\eta' \in (0, 1)$, the set $S_{l}$ satisfies that
    \begin{equation*}
        f(S_{l}) \geq \big(1 - e^{-\frac{l}{k}}\big) f(S^{*}_{k}) - \frac{4 k \Delta \ln(\frac{n}{\eta'})}{\varepsilon}
    \end{equation*}
    with probability at least $1 - l \eta'$, where $n = |V|$ and $S^{*}_{k} = \arg\max_{S \subseteq V : |S| \leq k} f(S)$.
\end{lemma}
\begin{proof}
    First of all, we consider the process of using the exponential mechanism to select the element $\widetilde{v}_i$ at the $i$-th iteration in Algorithm~\ref{alg:greedy}.
    By their definitions, we have $\Delta u_i = \max_{D \sim D'} |u_i(D, v) - u_i(D', v)|$ and $u_i(D, v) := f(S_{i - 1} \cup \{v\} | D) - f(S_{i-1} | D)$.
    Thus, we can obtain that $\Delta u_i \leq 2 \Delta$ for any $i \in [l]$.
    According to Eq.~\eqref{eq:em}, we have
    \begin{equation*}
        u_i(D, \widetilde{v}_i) \geq \max_{v \in V \setminus S_{i-1}} u_i(D, v) - \frac{4 \Delta}{\varepsilon}\ln(\frac{n}{\eta'})
    \end{equation*}
    with probability at least $1 - \eta'$.
    Next, we analyze the marginal gain of adding $\widetilde{v}_i$ to $S_{i - 1}$ at the $i$-th iteration compared to $S^{*}_{k}$.
    Define the marginal gain function $g(S, T) = f(S \cup T) - f(S)$ for any $S, T \subseteq V$.
    Obviously, $g(S_{i - 1}, \{v\})$ $= u_i(D, v)$.
    As such, we have
    \begin{align}
    \label{dpgreedy1}
        g(S_{i - 1}, \widetilde{v}_i) & \geq \max_{v \in V \setminus S_{i-1}} g(S_{i-1}, v) - \xi \nonumber\\
                                      & \geq \frac{1}{k} \Big(\sum_{v \in S^{*}_{k}} g(S_{i-1}, v) \Big) - \xi \nonumber\\
                                      & \geq \frac{f(S^{*}_{k} \cup S_{i-1}) - f(S_{i-1})}{k} - \xi \nonumber\\
                                      & \geq \frac{f(S^{*}_{k}) - f(S_{i-1})}{k} - \xi,
    \end{align}
    where $\xi = \frac{4 \Delta}{\varepsilon}\ln(\frac{n}{\eta'})$, the second inequality holds since the element to compare has the maximum marginal gain w.r.t.~$S_{i - 1}$, the third inequality holds from the submodularity of $f(\cdot)$, and the fourth inequality holds from the monotonicity of $f(\cdot)$.
    Then, we have
    \begin{equation}
    \label{dpgreedy2}
        f(S_i) = f(S_{i-1}) + (f(S_{i-1} \cup \{\widetilde{v}_i \}) - f(S_{i-1}))
        = f(S_{i-1}) + g(S_{i - 1}, \widetilde{v}_i).
    \end{equation}
    By substituting $g(S_{i - 1}, \widetilde{v}_i)$ in Eq.~\eqref{dpgreedy2} with its lower bound in Eq.~\eqref{dpgreedy1}, we have
    \begin{equation}
    \label{dpgreedy3}
        f(S_i) \geq f(S_{i-1}) + \frac{f(S^{*}_{k}) - f(S_{i-1})}{k} - \xi.
    \end{equation}
    By rearranging the terms in Eq.~\eqref{dpgreedy3}, we have
    \begin{equation*}
        f(S^{*}_{k}) - f(S_i) \leq \Big(1 - \frac{1}{k} \Big)\Big(f(S^{*}_{k}) - f(S_{i-1}) \Big) + \xi.
    \end{equation*}
    By expanding the above inequality from the base case of $f(S_0) = 0$, we have
    \begin{align*}
        f(S^{*}_{k}) - f(S_1) & \leq (1 - \frac{1}{k}) f(S^{*}_{k}) + \xi\\
        f(S^{*}_{k}) - f(S_2) & \leq (1 - \frac{1}{k}) (f(S^{*}_{k}) - f(S_1)) + \xi\\
                              & \leq (1 - \frac{1}{k})^2 f(S^{*}_{k}) + (1 + (1-\frac{1}{k})) \xi
    \end{align*}
    By repeating the above induction process from $j = l$ to $1$, we have
    \begin{equation}
    \label{dpgreedy5}
        f(S^{*}_{k}) - f(S_{l}) \leq \big(1-\frac{1}{k}\big)^{l} f(S^{*}_{k})
        + \xi \cdot \sum_{j = 0}^{l - 1} \big(1 - \frac{1}{k} \big)^j.
    \end{equation}
    In Eq.~\eqref{dpgreedy5}, the term $\sum_{j = 0}^{l - 1} (1 - \frac{1}{k})^j$ is the summation of a finite geometric series and computed as
    \begin{equation*}
        \sum_{j=0}^{l-1}\big( 1-\frac{1}{k}\big)^j= k\Big(1 - \big(1 - \frac{1}{k}\big)^{l} \Big).
    \end{equation*}
    Accordingly, we have
    \begin{equation*}
        f(S^{*}_{k}) - f(S_{l}) \leq \big(1 - \frac{1}{k} \big)^{l} f(S^{*}_{k})
        + \xi k \Big(1 - \big(1 - \frac{1}{k}\big)^{l} \Big).
    \end{equation*}
    By rewriting the above inequality, we have
    \begin{equation*}
        f(S_{l}) \geq \Big(1 - \big(1 - \frac{1}{k}\big)^{l} \Big) \big(f(S^{*}_{k}) - \xi k \big).
    \end{equation*}
    Since $(1-\frac{1}{k})^{l} \geq e^{-\frac{l}{k}}$, the inequality is simplified as
    \begin{align*}
        f(S_{l}) & \geq \big(1 - e^{-\frac{l}{k}}\big) \big(f(S^{*}_{k}) - \xi k\big)\\
                  & \geq \big(1 - e^{-\frac{l}{k}}\big) f(S^{*}_{k}) - \frac{4k \Delta}{\varepsilon}\ln(\frac{n}{\eta'}).
    \end{align*}
    Finally, the above results hold only if Eq.~\eqref{eq:em} is satisfied w.r.t.~$\eta'$ over $l$ iterations.
    According to the union bound, the probability that Eq.~\eqref{eq:em} is satisfied w.r.t.~$\eta'$ over $l$ iterations is at least $1 - l\eta'$.
\end{proof}

Based on Lemma~\ref{lm:dpgreedy}, we can analyze the approximation factor and the privacy guarantee of \textsc{DP-MultiGreedy} in the following theorem.

\begin{theorem}
\label{thm:dpgreedy}
    Suppose $\{f_1, \dots, f_d\}: 2^{V} \mapsto \mathbb{R}_{\geq 0}$ are monotone submodular functions with sensitivity $\Delta = O(1)$.
    For any parameter $\varepsilon > 0$, \textnormal{\textsc{DP-MultiGreedy}} provides $\varepsilon$-DP.
    Moreover, for any parameter $\eta \in (0, 1)$, with probability at least $1 - \eta$, its output solution $S$ satisfies
    \begin{equation*}
        F(S) \geq \left(1 - e^{-\frac{\lfloor k/d \rfloor}{k}}\right) F(OPT) - \xi,
    \end{equation*}
    where $F(S) = \min_{j \in [d]} f_j(S)$ and $\xi = O\!\left( \frac{d k}{\varepsilon} \ln\!\left(\frac{n k}{\eta}\right) \right).$
\end{theorem}
\begin{proof}
    The privacy guarantee of \textsc{DP-MultiGreedy} follows from the composition theorem.

    In the first phase, \textsc{DP-MultiGreedy} invokes \textsc{DP-Greedy} once for each objective function $f_j$, where $j \in [d]$, with privacy budget $\varepsilon_1/d$.
    Since all these executions access the same dataset $D$, by sequential composition, the first phase satisfies $\varepsilon_1$-DP.

    If the initial solution obtained after the first phase already has size $k' = k$, then the second phase is skipped.
    In this case, \textsc{DP-MultiGreedy} provides $\varepsilon_1$-DP, and hence also $\varepsilon$-DP since $\varepsilon_1 \leq \varepsilon$.
    Otherwise, if $k' < k$, the second phase performs $k-k'$ incremental iterations.
    In each iteration, the algorithm first estimates the current objective values using the Laplace mechanism and then selects a new element using the exponential mechanism.
    The privacy budget consumed in each iteration is $O\!\left(\frac{\varepsilon_2}{k-k'}\right)$.
    Therefore, by sequential composition, the second phase satisfies $\varepsilon_2$-DP.
    Combining the two phases, \textsc{DP-MultiGreedy} provides $(\varepsilon_1 + \varepsilon_2) = \varepsilon$-DP.

    Next, we analyze the approximation guarantee of \textsc{DP-MultiGreedy}.
    Let $S_j$ denote the subset generated for the objective function $f_j$ in the first phase.
    By construction, each $S_j$ has size $k_j = \lfloor \frac{k}{d} \rfloor$.
    According to Lemma~\ref{lm:dpgreedy}, for each $j \in [d]$, with probability at least $1 - k_j \eta'$, we have
    \begin{equation*}
        f_j(S_j)
        \geq
        \left(1 - e^{-k_j/k}\right) f_j(S^*_{j,k}) - \xi_0,
    \end{equation*}
    where
    \begin{equation*}
        S^*_{j,k} = \arg\max_{S \subseteq V : |S| \leq k} f_j(S),
    \end{equation*}
    and
    \begin{equation*}
        \xi_0 = \frac{4 d k \Delta}{\varepsilon_1} \ln\!\left(\frac{n}{\eta'}\right).
    \end{equation*}

    By the definition of the multi-objective value function,
    \begin{equation*}
        F(OPT) = \min_{j \in [d]} f_j(OPT).
    \end{equation*}
    For every $j \in [d]$, since $S^*_{j,k}$ is the optimal solution for maximizing $f_j$ under the cardinality constraint $k$, we have
    \begin{equation*}
        f_j(OPT) \leq f_j(S^*_{j,k}).
    \end{equation*}
    Hence,
    \begin{equation*}
        F(OPT) \leq f_j(S^*_{j,k}), \qquad \forall j \in [d].
    \end{equation*}
    Substituting this into the above inequality yields
    \begin{equation*}
        f_j(S_j)
        \geq
        \left(1 - e^{-\frac{\lfloor k/d \rfloor}{k}}\right) F(OPT) - \xi_0,
        \qquad \forall j \in [d].
    \end{equation*}

    Let $S_{\mathrm{init}} = \bigcup_{j=1}^d S_j$ be the initial solution constructed after the first phase.
    Since each $f_j$ is monotone and $S_j \subseteq S_{\mathrm{init}}$, we have
    \begin{equation*}
        f_j(S_{\mathrm{init}}) \geq f_j(S_j),
        \qquad \forall j \in [d].
    \end{equation*}
    Moreover, the final solution $S$ returned by the algorithm satisfies $S_{\mathrm{init}} \subseteq S$, and thus, by monotonicity again,
    \begin{equation*}
        f_j(S) \geq f_j(S_{\mathrm{init}}) \geq f_j(S_j),
        \qquad \forall j \in [d].
    \end{equation*}
    Therefore,
    \begin{align*}
        F(S)
        &= \min_{j \in [d]} f_j(S) \\
        &\geq \min_{j \in [d]} f_j(S_j) \\
        &\geq \left(1 - e^{-\frac{\lfloor k/d \rfloor}{k}}\right) F(OPT) - \xi_0.
    \end{align*}

    Finally, by the union bound, the above inequalities for all $j \in [d]$ hold simultaneously with probability at least
    \begin{equation*}
        1 - \sum_{j=1}^d k_j \eta'
        \geq 1 - k \eta'.
    \end{equation*}
    Letting $\eta = k \eta'$, and substituting $\Delta = O(1)$ and $\varepsilon_1 = \Theta(\varepsilon)$ into $\xi_0$, we obtain
    \begin{equation*}
        \xi = O\!\left( \frac{d k}{\varepsilon} \ln\!\left(\frac{n k}{\eta}\right) \right).
    \end{equation*}
    This completes the proof.
\end{proof}

The time complexity of \textsc{DP-MultiGreedy} is $O(kmn)$.
In the first phase, the algorithm invokes \textsc{DP-Greedy} for each of the $d$ objective functions, and each invocation uses a cardinality constraint $k' = \lfloor k/d \rfloor$.
Each greedy process performs $k'$ iterations, and in each iteration, it evaluates the marginal gain of $O(n)$ candidate elements.
Each evaluation takes $O(m)$ time on a dataset of size $m$.
Hence, the total time complexity of the first phase is $O(dk'nm)$.
In the second phase, the algorithm performs at most $k-k''$ iterations, where $k'' = |S|$ after the first phase.
In each iteration, the algorithm first evaluates the current values of the $d$ objective functions, which takes $O(dm)$ time, and then applies the exponential mechanism over the remaining $O(n)$ candidate elements, requiring $O(nm)$ time for marginal-gain evaluation.
Therefore, the time complexity of the second phase is $O((k - k'')nm)$.
Thus, the total time complexity of \textsc{DP-MultiGreedy} is $O(knm)$.

\subsection{The DP-Bicriteria Algorithm}
\label{subsec:alg2}

The \textsc{DP-MultiGreedy} typically shows inferior performance for a larger $d$ and becomes totally inapplicable when $d > k$.
Therefore, we propose a more general and improved algorithm called \textsc{DP-Bicriteria} that extends the \textsc{Saturate} algorithm in~\cite{krause2008robust} to satisfy DP.

\begin{algorithm}[t]
\caption{\textsc{DP-Bicriteria}}
\label{alg:bicriteria}
\begin{algorithmic}[1]
    \Require Ground set $V$, dataset $D$, set functions $f_1, \dots, f_d: 2^V$ $\mapsto \mathbb{R}_{\geq 0}$, cardinality constraint $k \geq 2$, privacy budgets $\varepsilon > 0$ and $\varepsilon_1, \varepsilon_2 > 0$ with $\varepsilon = \varepsilon_1 + \varepsilon_2$, parameters $\tau > 0$, $\eta \in (0, 1)$, and $\alpha \in (0, 1)$, search bounds $c_{\text{min}} < c_{\text{max}}$
    \Ensure Solution set $S_{\text{best}} \subseteq V$ with $|S_{\text{best}}| \leq k \ln(\frac{d}{\alpha})$
    \State Initialize $S_{\text{best}} = \emptyset$;
    \State $\gamma \gets \lceil \log_2(\frac{c_{\text{max}} - c_{\text{min}}}{\tau}) \rceil$;
    \State $\xi_0 = \frac{4 \gamma k \Delta}{\varepsilon_1}\ln\big(\frac{n \gamma}{\eta} (1 + k \ln(\frac{d}{\alpha}))\big) + \frac{\gamma \Delta}{\varepsilon_2} \ln\big(\frac{\gamma}{\eta} (1 + k \ln(\frac{d}{\alpha}))\big)$;
    \While{$c_{\text{max}} - c_{\text{min}} \geq \tau$}
        \State $c \gets (c_{\text{min}} + c_{\text{max}})/2$;
        \State Define $F_c(S) := \frac{1}{d} \sum_{i=1}^{d} \min(f_i(S), c)$;
        \State $S_c \gets \mathcal{G}(F_c, D, V, k \ln(\frac{d}{\alpha}), \frac{\varepsilon_1}{\gamma})$;
        \State $\widehat{F}_c(S_c) \gets \mathcal{M}_L(F_c(S_c), D, \frac{\varepsilon_2}{\gamma})$;
        \If{$\widehat{F}_c(S_c) \geq c(1 - \frac{\alpha}{d}) - \xi_0$}
            \State $c_{\text{min}} \gets c$ and $S_{\text{best}} \gets S_c$;
        \Else
            \State $c_{\text{max}} \gets c$;
        \EndIf
    \EndWhile
    \State \Return{$S_{\text{best}}$};
\end{algorithmic}
\end{algorithm}

\noindent
\textbf{Algorithm Description:}
We present the procedure of \textsc{DP-Bicriteria} in Algorithm~\ref{alg:bicriteria}.
We assume that the lower and upper bounds of $F(OPT)$ (i.e., $c_{\text{min}}$ and $c_{\text{max}}$) are known in advance.
Then, it performs a binary search in $[c_{\text{min}}, c_{\text{max}}]$ until $c_{\text{max}} - c_{\text{min}}$ is below the pre-specified $\tau > 0$.
Thus, there are at most $\gamma = \lceil \log_2(\frac{c_{\text{max}} - c_{\text{min}}}{\tau}) \rceil$ steps in the binary search process.
In each step, it sets the threshold $c$ as the midpoint of the current $c_{\text{min}}$ and $c_{\text{max}}$.
Subsequently, it formulates an objective function $F_c(S) := \frac{1}{d} \sum_{i=1}^{d} \min(f_i(S), c)$.
It is known from~\cite{krause2008robust} that $F_c(\cdot)$ is a monotone submodular function as long as $f_i(\cdot)$ is monotone submodular for all $i \in [d]$.
It is also easy to see that $F_c(S) = c$ only if $f_i(S) \geq c$ for all $i \in [d]$ and thus $F(S) \geq c$.
Therefore, the problem in each step turns into deciding whether there exists a subset $S \subseteq V$ with $|S| = k$ such that $F_c(S) = c$.
This decision problem can be treated as a single-objective submodular maximization problem for $F_c(\cdot)$ by relaxing the cardinality constraint $k$ to $k \ln(\frac{d}{\alpha})$.
The rationale for such a relaxation will be explained in a theoretical analysis later.
It then invokes \textsc{DP-Greedy} in Algorithm~\ref{alg:greedy} with a privacy budget $\frac{\varepsilon_1}{\gamma}$ to compute a solution $S_c$ of size $k \ln(\frac{d}{\alpha})$ for the maximization of $F_c$.
After that, it first uses the Laplace mechanism with a privacy budget $\frac{\varepsilon_2}{\gamma}$ to calculate a perturbed estimate $\widehat{F}_c(S_c)$ of $F_c(S_c)$.
If $\widehat{F}_c(S_c)$ is at least $c(1 - \frac{\alpha}{d}) - \xi_0$, where $\xi_0 = \frac{4 \gamma k \Delta}{\varepsilon_1}\ln\big(\frac{n \gamma}{\eta} (1 + k \ln(\frac{d}{\alpha}))\big) + \frac{\gamma \Delta}{\varepsilon_2} \ln\big(\frac{\gamma}{\eta} (1 + k \ln(\frac{d}{\alpha}))\big)$ is an error term due to DP noise, $S_c$ will be set as the current $S_{\text{best}}$ and $c_{\text{min}}$ will be updated to $c$ to search the upper half for possibly better solutions.
Otherwise, $c_{\text{max}}$ will be updated to $c$ to search the lower half for feasible solutions.
Finally, after the binary search process, the found $S_{\text{best}}$ will be returned for MOSM.

\noindent
\textbf{Theoretical Analysis:}
Next, we analyze the privacy, approximation, and time and space complexity of \textsc{DP-Bicriteria} in Algorithm~\ref{alg:bicriteria}.
First, we show the properties of the solution $S_c$ w.r.t.~$c$ in each step in the following lemma.

\begin{lemma}
\label{lm:bicriteria}
    For the solution $S_c$ in Algorithm~\ref{alg:bicriteria}, the following results hold with probability at least $1 - \eta' (1 + k \ln(\frac{d}{\alpha}))$:
    \emph{(i)} if $\widehat{F}_c(S_c) \geq c (1 - \frac{\alpha}{d}) - (\xi_1 + \xi_2)$, then $F(S_c) \geq (1 - \alpha) c - d (\xi_1 + 2 \xi_2)$; \emph{(ii)} if $\widehat{F}_c(S_c) < c (1 - \frac{\alpha}{d}) - (\xi_1 + \xi_2)$, then $F(OPT) < c$, where $\xi_1 = \frac{4 \gamma k \Delta}{\varepsilon_1} \ln(\frac{n}{\eta'})$ and $\xi_2 = \frac{\gamma \Delta}{\varepsilon_2} \ln(\frac{1}{\eta'})$.
\end{lemma}
\begin{proof}
    First, the sensitivity $\Delta F_c$ of $F_c$ satisfies that
    \begin{equation*}
        \Delta F_c = \frac{1}{d} \sum_{i = 1}^{d} \Delta f_i \leq \Delta.
    \end{equation*}
    According to Eq.~\eqref{eq:lap}, we have
    \begin{equation}
    \label{eq:lap:fc}
        |\widehat{F}_c(S_c) - F_c(S_c)| \leq \frac{\gamma \Delta}{\varepsilon_2} \ln(\frac{1}{\eta'}) = \xi_2
    \end{equation}
    with probability at least $1 - \eta'$.
    
    If $\widehat{F}_c(S_c) \geq c (1 - \frac{\alpha}{d}) - (\xi_1 + \xi_2)$, then we have $F_c(S_c) \geq c (1 - \frac{\alpha}{d}) - (\xi_1 + 2\xi_2)$ with probability at least $1 - \eta'$ due to Eq.~\eqref{eq:lap:fc}.
    Since $\min(f_i(S_c), c) \leq c$, we further have for all $i \in [d]$,
    \begin{equation*}
        \frac{f_i(S_c)}{d} \geq c (1 - \frac{\alpha}{d}) - (\xi_1 + 2\xi_2) - \frac{d - 1}{d} \cdot c.
    \end{equation*}
    Thus, we can obtain that with probability at least $1 - \eta'$,
    \begin{equation*}
        F(S_c) = \min_i f_i(S_c) \geq (1 - \alpha) c - d (\xi_1 + 2\xi_2).
    \end{equation*}

    If $\widehat{F}_c(S_c) < c (1 - \frac{\alpha}{d}) - (\xi_1 + \xi_2)$, we have $F_c(S_c) < c (1 - \frac{\alpha}{d}) - \xi_1$ with probability at least $1 - \eta'$ due to Eq.~\eqref{eq:lap:fc}.
    Taking $l = k \ln(\frac{d}{\alpha})$ into Lemma~\ref{lm:dpgreedy}, we have
    \begin{align*}
        F_c(S_c) & \geq \big(1 - e^{-\frac{k \ln(\frac{d}{\alpha})}{k}}\big) F_c(S^{*}_{k}) - \frac{4 k \gamma \Delta}{\varepsilon_1} \ln(\frac{n}{\eta'})\\
                 & \geq (1 - \frac{\alpha}{d}) F_c(S^{*}_{k}) - \xi_1
    \end{align*}
    with probability at least $1 - k \eta' \ln(\frac{d}{\alpha})$.
    Combining the above results, we have
    \begin{equation*}
        (1 - \frac{\alpha}{d}) F_c(S^{*}_{k}) - \xi_1 < c (1 - \frac{\alpha}{d}) - \xi_1,
    \end{equation*}
    which implies that $F_c(S^{*}_{k}) < c$ with probability at least $1 - \eta' \big(1 + k \ln(\frac{d}{\alpha})\big)$.
    As $S^{*}_{k}$ is the optimal solution for maximizing $F_c(\cdot)$ under a cardinality constraint $k$, $F_c(S^{*}_{k}) < c$ must imply that $F_c(S) < c$ for any set $S \subseteq V$ of size at most $k$.
    Therefore, we have $F_c(OPT) < c$, there is some $i \in [d]$ such that $f_i(OPT) < c$, and $F(OPT) < c$ with probability at least $1 - \eta' \big(1 + k \ln(\frac{d}{\alpha})\big)$.
\end{proof}

Based on Lemma~\ref{lm:bicriteria}, we analyze the approximation factor and the privacy guarantee of \textsc{DP-Bicriteria} in the following theorem.
\begin{theorem}
\label{thm:bicriteria}
    Suppose $f_1, \dots, f_d: 2^{V} \mapsto \mathbb{R}_{\geq 0}$ are monotone submodular functions with sensitivity $\Delta = O(1)$.
    For any parameter $\varepsilon >0$, \textnormal{\textsc{DP-Bicriteria}} provides $\varepsilon$-DP.
    Moreover, for any parameters $\alpha, \eta \in (0, 1)$, with probability at least $1 - \eta$, we have $|S_{\text{best}}| \leq k \ln(\frac{d}{\alpha})$ and $F(S_{\text{best}}) \geq (1 - \alpha) F(OPT) - \xi$, where
    $$\xi = O\Big( \frac{d k \log{m}}{\varepsilon} \log\big(\frac{k n \log{m}}{\eta} \log(\frac{d}{\alpha})\big) \Big).$$
\end{theorem}
\begin{proof}
    The privacy guarantee of \textsc{DP-Bicriteria} follows by applying the composition theorem.
    The algorithm runs in $\gamma$ iterations. In each iteration, it uses \textsc{DP-Greedy} with a privacy budget $\frac{\varepsilon_1}{\gamma}$ and the Laplace mechanism with a privacy budget $\frac{\varepsilon_2}{\gamma}$.
    Thus, the privacy budget of \textsc{DP-Bicriteria} is $\gamma (\frac{\varepsilon_1}{\gamma} + \frac{\varepsilon_2}{\gamma}) =\varepsilon$.

    For the approximation factor of \textsc{DP-Bicriteria}, we consider the final values of $c_{\text{min}}$ and $c_{\text{max}}$ when the binary search process terminates.
    Specifically, we have $c_{\text{max}} - c_{\text{min}} < \tau$, $\widehat{F}_{c_{\text{min}}}(S_{c_{\text{min}}}) \geq c_{\text{min}} (1 - \frac{\alpha}{d}) - (\xi_1 + \xi_2)$, and $\widehat{F}_{c_{\text{max}}}(S_{c_{\text{max}}}) < c_{\text{max}} (1 - \frac{\alpha}{d}) - (\xi_1 + \xi_2)$.
    Moreover, $c_{\text{min}}$ corresponds to $S_{\text{best}}$ returned by Algorithm~\ref{alg:bicriteria}.
    According to Lemma~\ref{lm:bicriteria}, we have $F(S_{\text{best}}) \geq (1 - \alpha) c_{\text{min}} - d (\xi_1 + 2 \xi_2)$ and $F(OPT) < c_{\text{max}}$.
    Note that the above results can hold only if the results in Lemma~\ref{lm:bicriteria} are satisfied across all $\gamma$ iterations, with a probability of at least $1 - \eta' \gamma (1 + k \ln(\frac{d}{\alpha}))$.
    By setting $\eta = \eta' \gamma (1 + k \ln(\frac{d}{\alpha}))$, we can obtain that
    \begin{equation*}
        F(S_{\text{best}}) \geq (1 - \alpha) F(OPT) - \tau - d (\xi_1 + 2 \xi_2)
    \end{equation*}
    with probability at least $1 - \eta$, where $\xi_1 = \frac{4 \gamma k \Delta}{\varepsilon_1} \cdot$ $\ln\big(\frac{n \gamma}{\eta} (1 + k \ln(\frac{d}{\alpha}))\big)$ and $\xi_2 = \frac{\gamma \Delta}{\varepsilon_2} \ln\big(\frac{\gamma}{\eta} (1 + k \ln(\frac{d}{\alpha}))\big)$.
     
    Finally, since it is reasonable to assume that $c_{\text{max}} - c_{\text{min}} = O(m)$ and $\tau = O(1)$, we have $\gamma = O(\log{m})$.
    As $\varepsilon_1, \varepsilon_2 = O(\varepsilon)$ and $\Delta = O(1)$, we can simplify $\xi = \tau + d (\xi_1 + 2 \xi_2)$ as
    \begin{equation*}
        \xi = O\Big( \frac{d k \log{m}}{\varepsilon} \log\big(\frac{k n \log{m}}{\eta} \log(\frac{d}{\alpha})\big) \Big),
    \end{equation*}
    and the proof is concluded.
\end{proof}

Finally, we analyze the time complexity of \textsc{DP-Bicriteria}.
It runs in $\gamma = O(\log{m})$ iterations, each of which takes $O\big(n m d k \log(\frac{d}{\alpha})\big)$ time for \textsc{DP-Greedy} and $O(md)$ time for the Laplace mechanism.
Therefore, its total time complexity is $O\big(n m d k \log{m}$ $\log(\frac{d}{\alpha})\big)$.

\section{Experiments}
\label{sec:exp}

In this section, we conduct extensive numerical experiments to evaluate the performance of the proposed \textsc{DP-MultiGreedy} and \textsc{DP-Bicriteria} algorithms for two representative submodular maximization applications, namely \emph{maximum coverage} and \emph{facility location}, in the multi-objective setting.

\subsection{Experimental Setup}
\label{subsec:setup}

Since we do not notice any existing DP algorithms for MOSM in the literature, we compare \textsc{DP-MultiGreedy} and \textsc{DP-Bicriteria} with the following non-private algorithms for MOSM in the experiments.
\begin{itemize}
    \item \textsc{GeneralizedGreedy} \cite{OrlinSU18} is an algorithm with an approximation factor of $1 - 1/e - \alpha$ ($0 < \alpha < 1$) for MOSM when $d = 2$.
    \item \textsc{MultiGreedy} is the non-private version of \textsc{DP-MultiGreedy}, which has an approximation factor of $(1 - e^{-\frac{\lfloor k/d \rfloor}{k}})$.
    \item \textsc{Saturate} \cite{krause2008robust} is a bicriteria $\big(1-\alpha, O(\log(\frac{d}{\alpha}))\big)$-approximation algorithm for MOSM.
    \item \textsc{MWU} \cite{Udwani18,FuB0P21} is a $((1 - 1/e)^2 - \alpha)$-approximation algorithm for MOSM based on multiplicative weight updates for any $d = o(\tfrac{k}{\log^3{k}})$.
\end{itemize}
All algorithms were implemented in Python 3.
All experiments were carried out on a server with an Intel Xeon E5-2650v4 processor @2.2GHz and 64GB of memory running Ubuntu 18.04.
For a fair comparison, \textsc{Saturate} and \textsc{DP-Bicriteria} were restricted to return size-$k$ solutions by replacing the term $k \ln(\frac{d}{\alpha})$ with $k$.
In the experiments, each algorithm used a single thread and was executed ten times with different random seeds.
Each evaluation measure was calculated by taking the average over ten runs.
To test the effect of $\varepsilon$ on DP algorithms, we further varied $\varepsilon \in \{0.1, 0.2, 0.4, \dots, 2.0\}$.
For each value of $\varepsilon$, we fixed $\varepsilon_1 = 0.9 \varepsilon$ and $\varepsilon_2 = 0.1 \varepsilon$ in \textsc{DP-MultiGreedy} and $\varepsilon_1 = 0.75 \varepsilon$ and $\varepsilon_2 = 0.25 \varepsilon$ in \textsc{DP-Bicriteria} because such privacy budget assignments achieved good empirical performance in most cases.

\subsection{Experimental Results for Maximum Coverage}
\label{subsec:mc}

\noindent\textbf{Setup:}
Maximum Coverage (MC) is a classic NP-hard combinatorial optimization problem.
Given a dataset $D$ of $m$ tuples and a ground set $V$ of $n$ sets, where each $V' \in V$ is a subset of $D$ (i.e.,  $V' \subseteq D$), the goal of MC is to select a subset $S \subseteq V$ of size $k$ such that the union of the sets in $S$ contains the maximum number of tuples in $D$.
The utility function for MC is $f(S) = \lvert \bigcup_{V' \in S} V' \rvert$, which is known to be normalized, monotone, and submodular.
By applying DP, any single tuple $x \in D$ is less distinguishable from the output of the algorithm.
In the multi-objective setting, we define the $i$-th objective function $f_i$ with an $m$-dimensional vector $\mathbf{w}_i$ to indicate whether the $j$-th tuple $x_j \in D$ is used in the calculation of $f_i$ (if $w_{ij} = 1$) or not (if $w_{ij} = 0$).
That is, we define $f_i(S) = \lvert \{ j \in [m] \,:\, x_j \in \bigcup_{V' \in S} V' \wedge w_{ij} = 1 \} \rvert$, which is still normalized, monotone, and submodular.
The sensitivity $\Delta f_i$ of each $f_i$ is $1$ and thus $\Delta = 1$.

We use two public datasets in the experiments for maximum coverage.
The DBLP dataset \cite{Ley09} is a bipartite network with $m = 704,738$ researchers (tuples) and $n = 2,675$ venues (sets).
The set for each venue consists of all researchers who published at least one paper on the venue between 2018 and 2021.
The Flickr dataset \cite{MisloveMGDB07} is also a bipartite network with $m = 242,364$ users (tuples) and $m = 2,000$ groups (sets).
The set for each group consists of all users who joined the group.
We randomly generate $d$ subsets of equal size $\lfloor\frac{m}{d}\rfloor$ from $D$ for the definition of each $f_i$.

\begin{table}[t]
    \centering
    \caption{Performance of each algorithm for maximum coverage in terms of the utility ratio (UR, the ratio of the value of $F(S)$ by each algorithm to the highest one among all algorithms) and the running time (in seconds) when $k = 5$, $\varepsilon = 1, 2$ and $d = 2, 5$. The best URs without and with DP are highlighted in bold and italic bold.}
    \label{tab:comparison1}
    \begin{tabular}{|c|cc|cc|cc|cc|}
        \hline
        \multirow{3}{*}{Algorithm}
        & \multicolumn{4}{c|}{$d=2$}
        & \multicolumn{4}{c|}{$d=5$} \\
        \cline{2-9}
        & \multicolumn{2}{c|}{DBLP}
        & \multicolumn{2}{c|}{Flickr}
        & \multicolumn{2}{c|}{DBLP}
        & \multicolumn{2}{c|}{Flickr} \\
        \cline{2-9}
        & UR & Time(s) & UR & Time(s) & UR & Time(s) & UR & Time(s) \\
        \hline
        \textsc{GeneralizedGreedy} & 0.71       & 2.18   & 0.70       & 4.76   & -    & -      & -    & -      \\
        \textsc{MultiGreedy}       & \textbf{1} & 150.0  & \textbf{1} & 214.9  & 0.77 & 142.31 & 0.80 & 159.77 \\
        \textsc{Saturate}          & \textbf{1} & 1540.4 & \textbf{1} & 2164.3 & \textbf{1} & 4999.2 & \textbf{1} & 4140.4 \\
        \textsc{MWU}               & 0.96          & 67.4   & 0.98          & 93.9   & 0.97 & 208.4  & 0.98 & 410.6  \\
        \hline
        \textsc{DP-MultiGreedy} ($\varepsilon = 1$) & \textbf{\textit{1}} & 141.9 & \textbf{\textit{1}} & 192.0 & 0.69 & 1616.9 & 0.71 & 3878.2 \\
        \textsc{DP-Bicriteria} ($\varepsilon = 1$)  & 0.93                   & 643.3 & 0.97                   & 972.4 & 0.74 & 1622.6 & 0.79 & 1334.9 \\
        \textsc{DP-MultiGreedy} ($\varepsilon = 2$) & \textbf{\textit{1}} & 129.5 & \textbf{\textit{1}} & 182.7 & 0.72 & 1629.6 & 0.75 & 3591.4 \\
        \textsc{DP-Bicriteria} ($\varepsilon = 2$)  & \textbf{\textit{1}} & 325.8 & \textbf{\textit{1}} & 211.3 & \textbf{\textit{0.94}} & 862.7 & \textbf{\textit{0.91}} & 675.8 \\
        \hline
    \end{tabular}
\end{table}

\noindent
\textbf{Results:}
We first present the performance of each algorithm in terms of solution quality and time efficiency when $k = 5$, $\varepsilon = 1, 2$ and $d = 2, 5$ in Table~\ref{tab:comparison1}.
When there are $d = 2$ objective functions, among the non-private algorithms, \textsc{GeneralizedGreedy} always provides the worst solutions, \textsc{MultiGreedy} and \textsc{Saturate} return the best solutions with an equal utility value, and \textsc{MWU} offers solutions with slightly lower utilities than the best ones.
Despite its high quality of solution, \textsc{Saturate} runs significantly slower than all other private algorithms due to its higher time complexity.
When $d = 5$, \textsc{Saturate} achieves slightly better performance than \textsc{MWU}$,$ while showing a clear advantage over \textsc{MultiGreedy}.
Among DP algorithms, \textsc{DP-MultiGreedy} shows superior performance when $d = 2$: it provides the same solutions as \textsc{MultiGreedy} and \textsc{Saturate} when $\varepsilon = 1$ while running much faster than \textsc{DP-Bicriteria}.
However, \textsc{DP-MultiGreedy} exhibits a significant performance decline when $d > 2$.
\textsc{DP-Bicriteria} can also return the same solutions as \textsc{MultiGreedy} and \textsc{Saturate} when $d = 2$ and $\varepsilon = 2$.
However, its solutions are slightly inferior to the best when $\varepsilon = 1$ or $d = 5$ due to the noise introduced by the DP mechanisms.
Finally, \textsc{DP-Bicriteria} runs faster than \textsc{Saturate} due to fewer steps in the binary search process.

\begin{figure}[t]
    \centering
    \includegraphics[width=.98\linewidth]{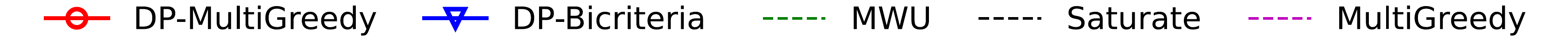}\\
    \vspace{1mm}
    \includegraphics[width=.67\linewidth]{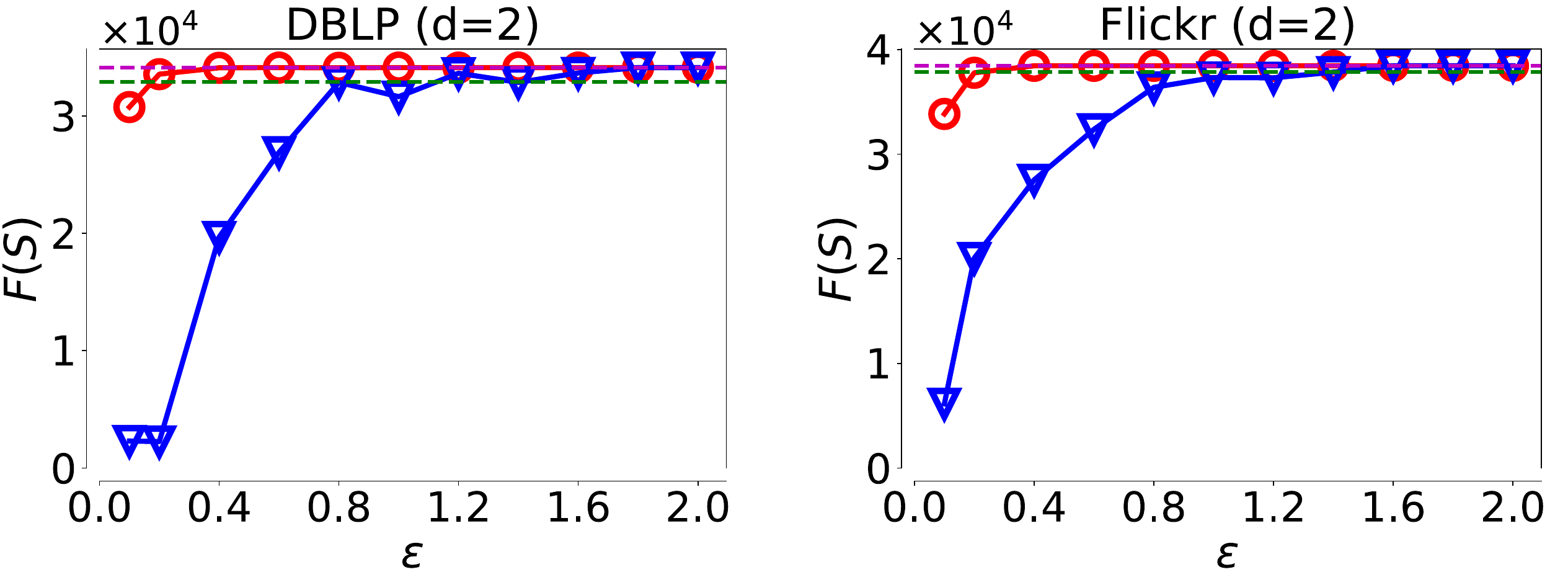}
    \vspace{1mm}
    \includegraphics[width=.67\linewidth]{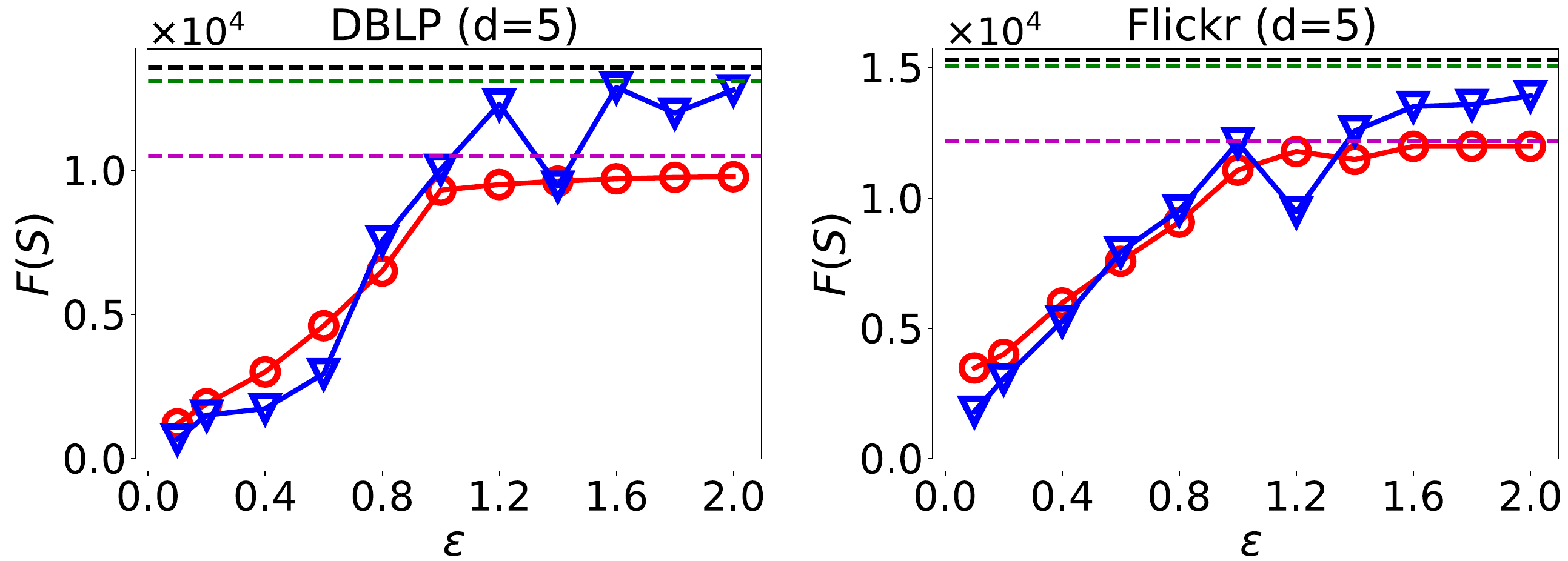}
    \caption{Solution quality of each algorithm for maximum coverage by varying the privacy budget $\varepsilon$ from $0.1$ to $2$.}
    \label{fig:eps-mc}
\end{figure}

Fig.~\ref{fig:eps-mc} illustrates the solution quality of each algorithm by varying the privacy budget $\varepsilon$ from $0.1$ to $2.0$.
When the value of $\varepsilon$ is smaller, the DP algorithm provides a higher level of privacy protection but introduces more noise to its solution.
We observe the strong performance of \textsc{DP-MultiGreedy} in the case of $d = 2$, which provides the same solutions as non-private baselines on both datasets when $\varepsilon \geq 0.4$.
When $d = 2$, \textsc{DP-Bicriteria} returns better solutions with increasing $\varepsilon$, converging to the best ones when $\varepsilon \geq 1.6$.
This indicates that \textsc{DP-Bicriteria} should use much more privacy budget than \textsc{DP-MultiGreedy}.
When $d = 5$, the solution quality of both \textsc{DP-Bicriteria} and \textsc{DP-MultiGreedy} generally increases with $\varepsilon$. However, the utility gap between \textsc{DP-MultiGreedy} and \textsc{DP-Bicriteria} becomes more pronounced in this case, and \textsc{DP-Bicriteria} remains consistently closer to the optimal utility. Although it does not fully converge to the non-private solution until $\varepsilon = 2$, it may still show occasional fluctuations due to randomness.

\begin{figure}[t]
    \centering
    \includegraphics[width=.98\linewidth]{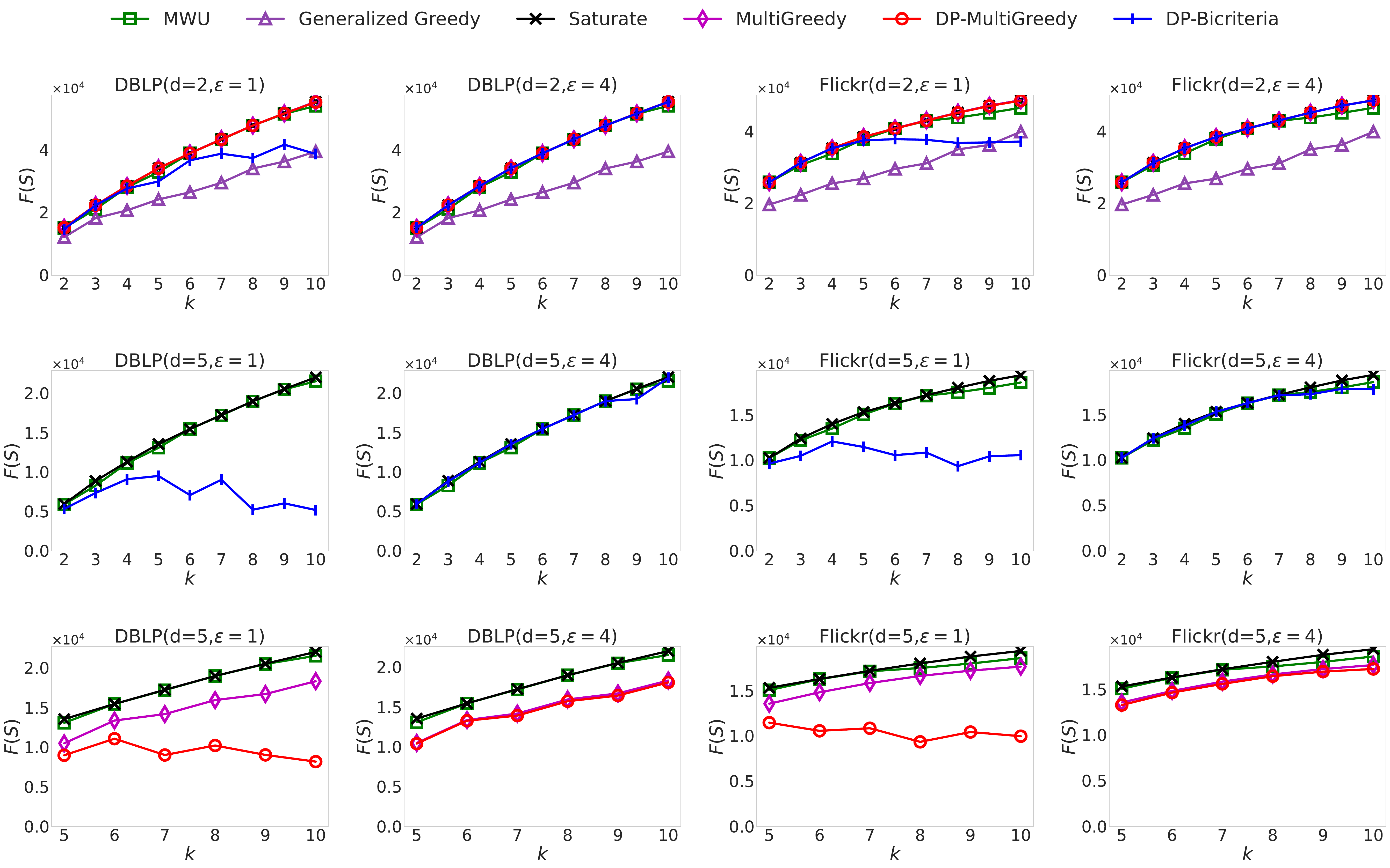}
    \caption{Solution quality of each algorithm for maximum coverage by varying the cardinality constraint $k$ from $2$ to $10$.}
    \label{fig:k-mc}
\end{figure}

Fig.~\ref{fig:k-mc} presents the solution quality of each algorithm as the cardinality constraint $k$ increases from $2$ to $10$.  The figure shows the performance of the two DP algorithms under $d=2$ and $d=5$, with privacy budgets $\varepsilon=1$ and $4$. Since \textsc{DP-MultiGreedy} requires $d \leq k$, we only report its results for $k=5$ to $10$ in the last group of subfigures when $d=5$. Overall, the utility values of all algorithms increase steadily with $k$, indicating that larger cardinality constraints lead to better solutions. Among the non-private algorithms, \textsc{Saturate}, \textsc{MultiGreedy}, and \textsc{MWU} achieve very similar solution quality across different values of $k$, while \textsc{GeneralizedGreedy} consistently performs worse. Among the DP algorithms, \textsc{DP-MultiGreedy} remains close to the non-private methods when $d=2$, even for $\varepsilon=1$, and almost matches them when $\varepsilon=4$. In contrast, the solution quality of \textsc{DP-Bicriteria} decreases somewhat as $k$ increases when $\varepsilon=1$, mainly due to the accumulation of privacy noise from more mechanism calls, but becomes close to the non-private baselines when the privacy budget is sufficiently large. When the number of objective functions increases to $d=5$, the overall trend remains similar, but the effect of privacy noise becomes more pronounced: for $\varepsilon=1$, both DP algorithms show a clear gap from the non-private methods, while for $\varepsilon=4$, their results become much closer to the non-private algorithms, indicating that a larger privacy budget is needed to maintain stable solution quality in higher-dimensional multi-objective settings.

\begin{figure}[t]
    \centering
    \includegraphics[width=.98\linewidth]{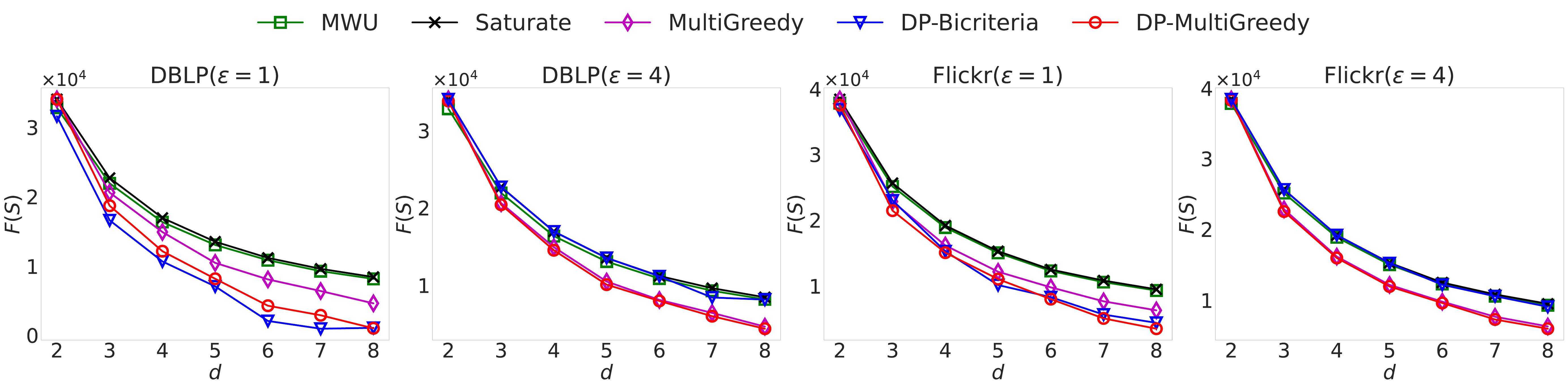}
    \caption{Solution quality of each algorithm for maximum coverage by varying the number of objective functions $d$ from $2$ to $8$.}
    \label{fig:d-mc}
\end{figure}

Fig.~\ref{fig:d-mc} illustrates the solution quality of each algorithm when the number of objective functions $d$ increases from $2$ to $8$. As the value of $d$ increases, the utility values of the solutions of all algorithms decrease because the number of tuples w.r.t.~each objective function is $\lfloor \frac{m}{d} \rfloor$. Compared with the other methods, \textsc{MultiGreedy} and \textsc{DP-MultiGreedy} exhibit a more pronounced performance drop as $d$ increases. Both DP algorithms with $\varepsilon = 4$ maintains performance comparable to that of non-private baselines across different $d$'s. When $\varepsilon = 1$, its performance is lower due to an insufficient privacy budget.

\subsection{Experimental Results for Facility Location}
\label{subsec:fl}

\noindent\textbf{Setup:}
Facility location (FL) is also a well-known NP-Hard combinatorial optimization problem.
We formulate FL as a submodular maximization problem according to Ref.~\cite{LindgrenWD17}.
Specifically, given a dataset $D$ of $m$ user locations in the 2D Euclidean space $\mathbb{R}^2$ and a ground set $V$ of $n$ candidate locations also in $\mathbb{R}^2$ for the deployment of facilities, the objective of FL is to select a subset $S \subseteq V$ of size $k$ to maximize the total benefits of $S$ w.r.t.~$D$.
We adopt the radial basis function (RBF) kernel to calculate the benefit of a candidate $v \in V$ for a user location $x \in D$, i.e., $b_{xv} = \exp( -\beta \| x - v \|_2^2)$, where $\beta > 0$ is a radius factor typically determined by the average distance between the locations in $V$ and $D$.
Obviously, we have $0 \leq b_{xv} \leq 1$ for any $v \in V$ and $x \in D$.
Then, the FL objective function is defined as $f(S) = \sum_{x \in D} \max_{v \in S} b_{xv}$, which is known to be normalized, monotone, and submodular.
By applying DP, any single user location $x \in D$ is less distinguishable from the output of the algorithm.
In the multi-objective setting, we also define the $i$-th objective function $f_i$ with an $m$-dimensional vector $\mathbf{w}_i$ to indicate whether the $j$-th tuple $x_j \in D$ is used in the calculation of $f_i$ (if $w_{ij} = 1$) or not (if $w_{ij} = 0$).
That is, we define $f_i(S) = \sum_{j \in [m]} w_{ij} \cdot \max_{v \in S} b_{x_j v}$, which is still normalized, monotone, and submodular.
The sensitivity $\Delta f_i$ of each $f_i$ is $1$ and thus $\Delta = 1$.

We also use two public datasets in the experiments for facility location.
The FourSquare dataset \cite{YangZZY15} consists of a set of 100,000 user check-in locations in Tokyo.
The Gowalla dataset \cite{ChoML11} contains a set of 138,954 user check-in locations in New York City.
For each dataset, we generate 256 candidate locations for the deployment of facilities by drawing a $16 \times 16$ grid evenly across the range of Tokyo/New York City.
We also randomly generate $d$ subsets of equal size $\lfloor\frac{m}{d}\rfloor$ from $D$ for the definition of $f_i$.

\begin{table}[t]
    \centering
    \caption{Performance of each algorithm for facility location in terms of the utility ratio (UR) and the running time (in seconds) when $k = 5$, $\varepsilon = 1, 2$ and $d = 2, 5$. The best values of $F(S)$ without and with DP are highlighted in bold and italic bold.}
    \label{tab:comparison2}
    \begin{tabular}{|c|cc|cc|cc|cc|}
        \hline
        \multirow{3}{*}{Algorithm}
        & \multicolumn{4}{c|}{$d=2$}
        & \multicolumn{4}{c|}{$d=5$} \\
        \cline{2-9}
        & \multicolumn{2}{c|}{FourSquare}
        & \multicolumn{2}{c|}{Gowalla}
        & \multicolumn{2}{c|}{FourSquare}
        & \multicolumn{2}{c|}{Gowalla} \\
        \cline{2-9}
        & UR & Time(s) & UR & Time(s) & UR & Time(s) & UR & Time(s) \\
        \hline
        \textsc{GeneralizedGreedy} & 0.66          & 38.14  & 0.80          & 70.95  & -    & -      & -    & -      \\
        \textsc{MultiGreedy}       & \textbf{1} & 143.3  & \textbf{1} & 214.2  & 0.87 & 53.9   & 0.97 & 78.9   \\
        \textsc{Saturate}          & \textbf{1} & 1061.8 & \textbf{1} & 1236.9 & \textbf{1} & 2811.1 & \textbf{1} & 4122.2 \\
        \textsc{MWU}               & 0.79          & 249.8  & 0.96          & 300.9  & 0.78 & 495.5  & 0.96 & 432.3  \\
        \hline
        \textsc{DP-MultiGreedy} ($\varepsilon = 1$)   & \textbf{\textit{1}}                   & 97.22  & \textbf{\textit{1}}                  & 148.2  & 0.83                   & 1587.50 & 0.88                   & 1566.8 \\
        \textsc{DP-Bicriteria} ($\varepsilon = 1$)    & 0.96                   & 865.1  & 0.99                   & 1138.5 & 0.85                   & 593.0   & 0.90                   & 991.9  \\
        \textsc{DP-MultiGreedy} ($\varepsilon = 2$)   & \textbf{\textit{1}} & 98.76  & \textbf{\textit{1}} & 147.6  & 0.86                   & 1587.50 & 0.95                   & 1565.50 \\
        \textsc{DP-Bicriteria} ($\varepsilon = 2$)    & 0.99                   & 237.6  & \textbf{\textit{1}}                  & 257.9  & \textbf{\textit{0.96}} & 252.8   & \textbf{\textit{0.98}} & 562.2  \\
        \hline
    \end{tabular}
\end{table}

\noindent\textbf{Results:}
Table~\ref{tab:comparison2} presents a performance comparison of different algorithms for facility location on the FourSquare and Gowalla datasets when $k = 5$, $\varepsilon = 1, 2$ and $d = 2, 5$.
Generally, we observe that the performance of each algorithm is consistent between the FL and MC problems.
Among non-private algorithms, \textsc{MultiGreedy} and \textsc{Saturate} still perform the best, \textsc{MWU} ranks third, and \textsc{GeneralizedGreedy} has the poorest solution quality.
\textsc{Saturate} also has the lowest time efficiency, whereas \textsc{MultiGreedy} is specific to the case of $d = 2$.
In addition, \textsc{DP-MultiGreedy} has slightly better quality of solutions than \textsc{DP-Bicriteria} under the same privacy budget when $d = 2$.
A notable distinction between the results for FL and MC emerges when $\varepsilon = 2$.
Although \textsc{DP-Bicriteria} achieves performance comparable to non-private algorithms for MC when $\varepsilon = 2$ and $d = 2$, its performance for FL falls slightly shorter than that of non-private algorithms.
This indicates that FL is a more challenging problem than MC in the context of DP.

\begin{figure}[t]
    \centering
    \includegraphics[width=.98\linewidth]{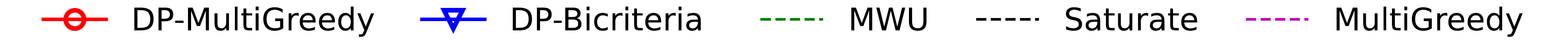}\\
    \vspace{1mm}
    \includegraphics[width=.67\linewidth]{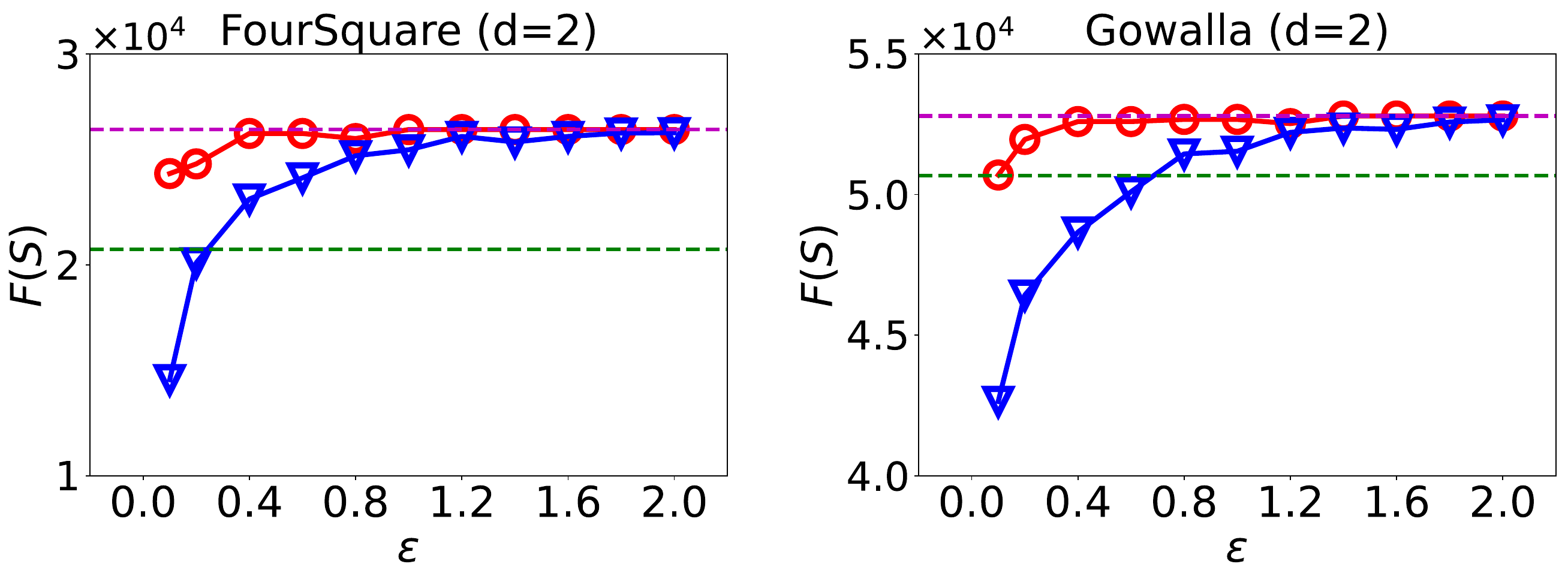}
    \vspace{1mm}
    \includegraphics[width=.67\linewidth]{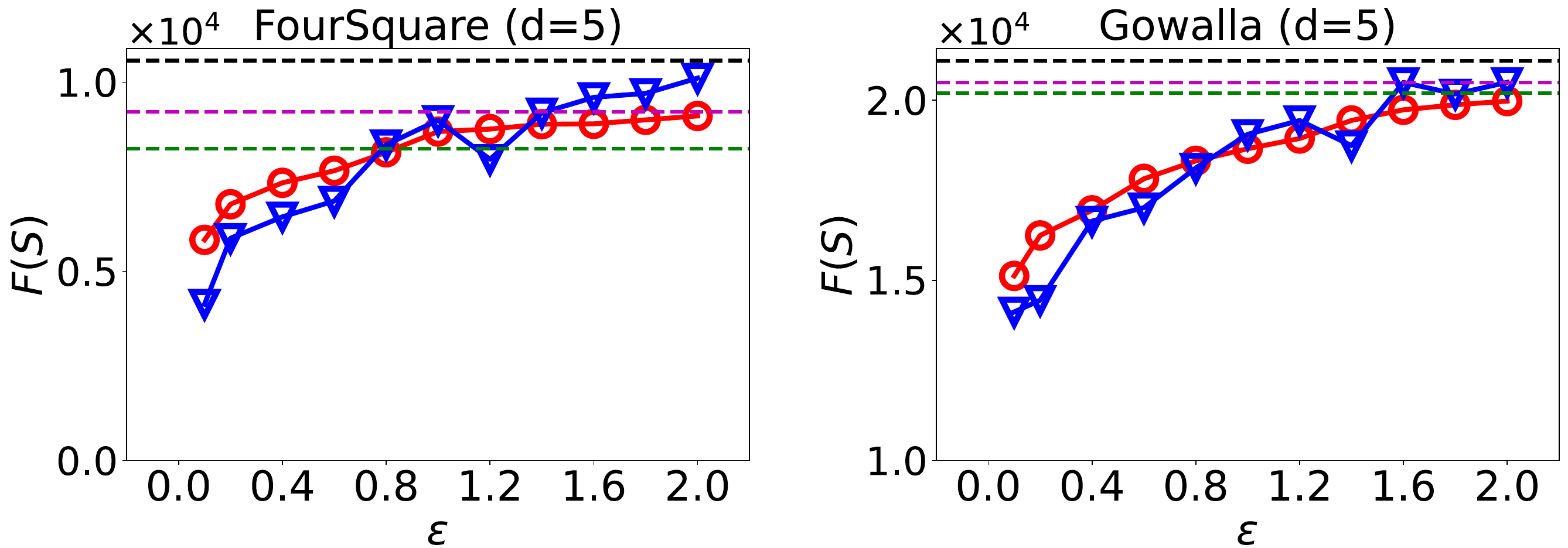}
    \caption{Solution quality of each algorithm for facility location by varying the privacy budget $\varepsilon$ from $0.1$ to $2$.} 
    \label{fig:eps-fl}
\end{figure}
\begin{figure}[!t]
    \centering
    \includegraphics[width=.98\linewidth]{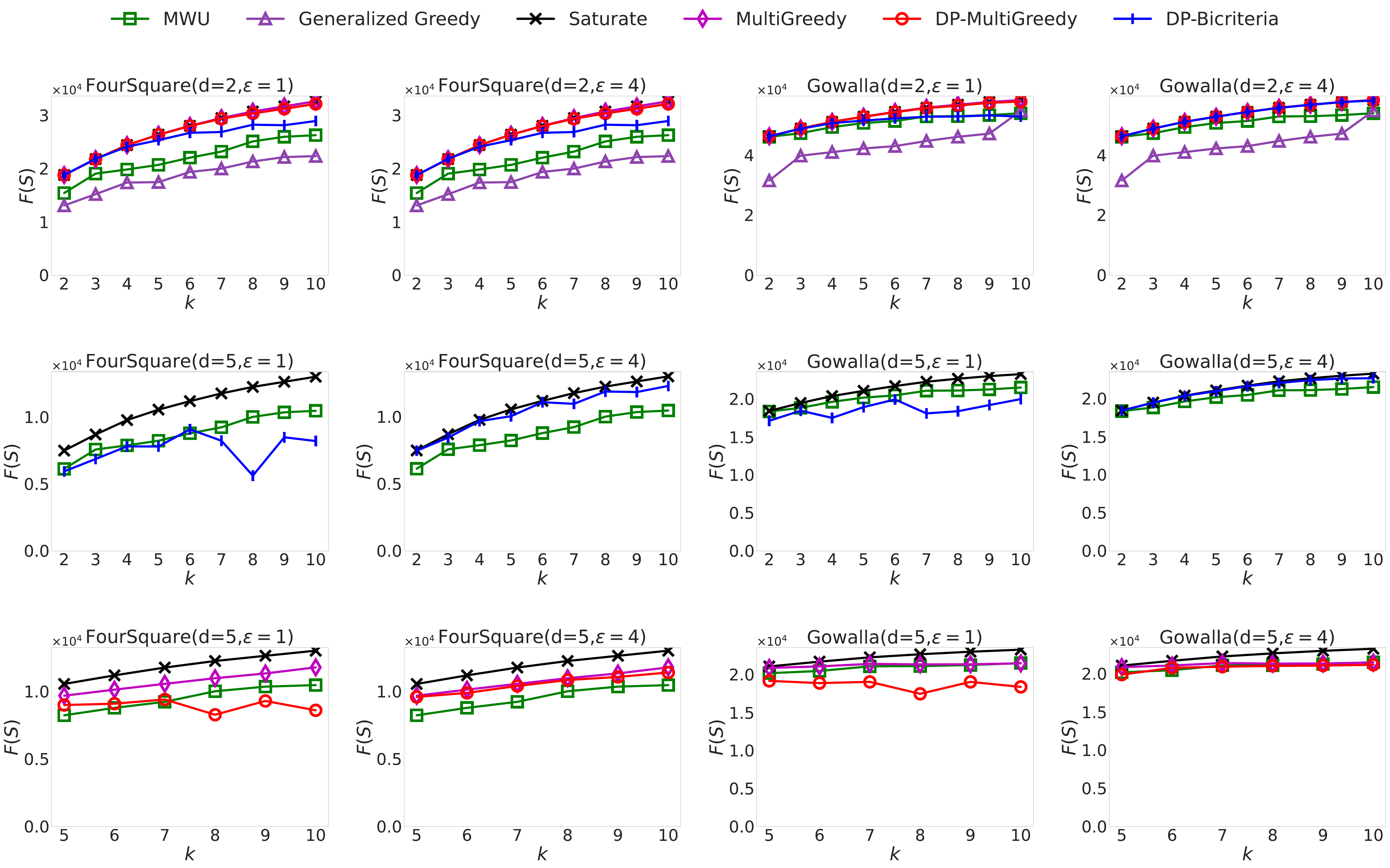}
    \caption{Solution quality of each algorithm for facility location by varying the cardinality constraint $k$ from $2$ to $10$.} 
    \label{fig:k-fl}
\end{figure}

Fig.~\ref{fig:eps-fl} shows how the privacy budget $\varepsilon$ affects the solution quality of each algorithm. We generally observe similar trends to those for MC. At lower $\varepsilon$ values, the DP algorithm achieves better privacy protection at the expense of worse solutions. In the case of $d = 2$, \textsc{DP-MultiGreedy} and \textsc{DP-Bicriteria} match the utilities of the best non-private baselines when $\varepsilon \geq 1.0$ and $1.8$, respectively. Both values of $\varepsilon$ are higher than their counterparts for MC (i.e., $\varepsilon \geq 0.4$ and $1.6$), which further confirms that FL requires more privacy budget than MC. For $d = 5$, the solution quality of both \textsc{DP-Bicriteria} and \textsc{DP-MultiGreedy} generally increases with $\varepsilon$, but \textsc{DP-MultiGreedy} still shows a more noticeable gap from the optimal results, whereas \textsc{DP-Bicriteria} achieves better utility overall, despite slight fluctuations caused by randomness.

Fig.~\ref{fig:k-fl} presents how the solution quality of each algorithm is affected by the cardinality constraint $k$.
We present the results of the two DP algorithms with $\varepsilon = 1$ and $4$ to keep them consistent with those of MC.
We also draw similar findings as for MC.
As $k$ increases, the utility values of all algorithms increase.
In the non-private setting, \textsc{Saturate} and \textsc{MultiGreedy} obtain the best solutions across different values of $k$.
However, we find that \textsc{MWU} returns solutions significantly inferior to \textsc{Saturate} and \textsc{MultiGreedy} in FL, indicating that \textsc{MWU} might not be appropriate for FL problems.
\textsc{DP-MultiGreedy} still achieves solution qualities on par with non-private algorithms for both $\varepsilon = 1$ and $4$.
The performance of \textsc{DP-Bicriteria} significantly decreases as $k$ increases for $\varepsilon = 1$ due to an insufficient budget.
However, in the case of $\varepsilon = 4$, its utility can align closely with that of non-private algorithms.

\begin{figure}[t]
    \centering
    \includegraphics[width=.98\linewidth]{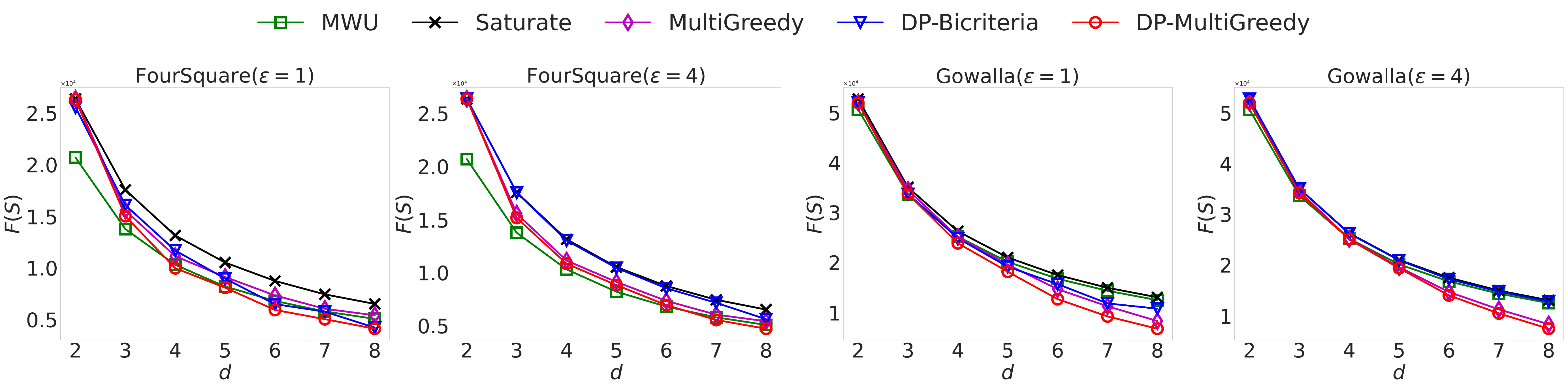}
    \caption{Solution quality of each algorithm for facility location by varying the number of objective functions $d$ from $2$ to $8$.}
    \label{fig:d-fl}
\end{figure}

Fig.~\ref{fig:d-fl} shows the quality of the solutions provided by each algorithm with different numbers of objective functions $d$.
We still observe similar patterns as for MC.
With an increase in $d$, the utility values of the solutions generated by all algorithms decrease, since the number of tuples corresponding to each objective function is $\lfloor \frac{m}{d} \rfloor$.
The distinction between the results for FL and MC is that \textsc{MWU} performs even worse than \textsc{DP-Bicriteria} in FL.

\section{Conclusion}
\label{sec:conclusion}

In this paper, we studied the cardinality-constrained multi-objective submodular maximization problem in the context of differential privacy.
We proposed two novel algorithms for the problem and analyzed their privacy guarantees, approximation factors, and time complexities.
Finally, we conducted extensive numerical experiments for two submodular maximization problems, maximum coverage and facility location, to demonstrate the effectiveness and efficiency of the proposed algorithms.

This paper still leaves open problems for future work.
How to extend our DP algorithms to handle more complex constraints, such as matroid and knapsack constraints, has not yet been considered.
In addition, it would also be interesting to introduce DP continuous greedy algorithms into multi-objective submodular maximization problems to improve the approximation factors.

\bibliographystyle{splncs04}
\bibliography{mybib}

\end{document}